\documentclass[]{interact}
\usepackage{subfigure,setspace,mathtools}
\usepackage{epstopdf,color}

\usepackage[numbers,sort&compress]{natbib}
\bibpunct[, ]{[}{]}{,}{n}{,}{,}

\theoremstyle{plain}
\newtheorem{theorem}{Theorem}[section]

\newtheorem{corollary}[theorem]{Corollary}

\theoremstyle{definition}

\theoremstyle{remark}
\newtheorem{remark}{Remark}

\newcommand\Mycomb[2][^n]{\prescript{#1\mkern-0.5mu}{}C_{#2}}

\begin{document}
	
	\articletype{Original Article}
	
	\title{A contagion process with self-exciting jumps in credit risk applications}
	
	\author{
		\name{Puneet  Pasricha\textsuperscript{a}, Dharmaraja Selvamuthu\textsuperscript{a} and  Selvaraju Natarajan\textsuperscript{b}\thanks{CONTACT Dharmaraja Selvamuthu. Email: dharmar@maths.iitd.ac.in} }
		\affil{\textsuperscript{a}              Department of Mathematics,               Indian Institute of Technology Delhi,
			Hauz Khas, New Delhi-110016, INDIA
		}
		\affil{\textsuperscript{b}Department of Mathematics, Indian Institute of Technology Guwahati,\\ Guwahati-781039, INDIA}
	}
	
	\maketitle
	
	\begin{abstract}
		The modeling of the probability of joint default or total number of defaults among the firms is one of the crucial problems to mitigate the credit risk since the default correlations significantly affect the portfolio loss distribution and hence play a significant role in allocating capital for solvency purposes. In this article, we derive a closed-form expression for the probability of default of a single firm and the probability of the total number of defaults by any time $t$ in a homogeneous portfolio of firms. We use a contagion process to model the arrival of credit events that causes the default and develop a framework that allows firms to have resistance against default unlike the standard intensity-based models. We assume the point process driving the credit events to be composed of a systematic and an idiosyncratic component, whose intensities are independently specified by a mean-reverting affine jump-diffusion process with self-exciting jumps. The proposed framework is competent of capturing the feedback effect, an empirically observed phenomenon in the default events. We further demonstrate how the proposed framework can be used to price synthetic collateralized debt obligation (CDO) and obtain a closed-form solution for tranche spread. Finally, we present the sensitivity analysis to demonstrate the effect of different parameters governing the contagion effect on the spread of tranches and the expected loss of the CDO.
	\end{abstract}
	
	\section{Introduction}\label{introduction}
	
	One of the significant concerns in credit risk management is the modeling of default dependence among the defaultable claims in a credit portfolio, for instance, a portfolio of loans. Dependence among the default events is vital for estimating the portfolio losses, which helps to allocate capital for solvency and pricing of credit derivatives, e.g., credit default swaps (CDS), collateralized debt obligations (CDOs), etc. There are two modeling frameworks in credit risk modeling, namely the firm-value models (also known as structural models) and the reduced-form models (also known as intensity based models). 
	
	In pricing of structured credit derivatives such as collateralized debt obligation using reduced-form models, there are mainly two ways to model the joint default risk or the portfolio credit risk, namely the top-down approach and the bottom-up approach. Most of the intensity based models in the literature assume that the default event is the first jump time of a point process with a stochastic intensity driven by a L$\acute{e}$vy process. Further, due to independent increment property of L$\acute{e}$vy processes, they are not capable of capturing the clustering phenomenon observed in the financial markets. However, as argued in A$\ddot{i}$t-Sahalia et al. \cite{ait2015} that ``what makes financial crisis take place is typically not the initial jump, but the amplification that occurs subsequently over hours or days, and the fact that other markets become affected as well". In other words, each firm has a resistance against adverse events, which implies that it is not the first adverse event that causes the firm to default but the amplification caused by these events over time.
	
	In this paper, we propose a reduced-form model by assuming that the defaults of individual firms are independent conditional on the macro-economic events.  We assume that the point process modeling the arrival of adverse events of the $i$th firm is a combination of two processes that models the idiosyncratic events and macro-economic events. We assume that the intensity of the point-processes follow a mean-reverting jump-diffusion process with clustering effects (\cite{dassios2011dynamic},  \cite{dassios2017generalized}   \cite{pasricha2019markov}). The proposed model is capable of modeling the contagious arrivals of events. The contribution of this article is two-fold. Firstly, the proposed model goes beyond the standard intensity-based models in the sense that we assume that each defaultable firm has some resistance to default, and the default is not the first jump time of the point process as assumed in most of the models in the literature (\cite{dassios2011dynamic}). Secondly, we present a parametric closed-form expression for the probability generating function (PGF) of the dynamic contagion process proposed in Dassios and Zhao \cite{dassios2011dynamic} which in turn yields a parametric closed form solution to joint probability of default. Finally, we study the application of our model to obtain the prices of tranches of CDOs in a bottom-up framework and perform numerical experiments to explore the impact of various parameters governing contagion on the prices.

	The organization of this article is as follows. Section 2 gives the literature review on modeling of joint distribution of defaults. Section 3 discusses the proposed modeling  framework and gives the methodology to calculate the joint distribution of multiple defaults. Section 4 presents application of the proposed model to obtain spread of tranches of a CDO and discusses the sensitivity analysis of the obtained prices against the parameters of the proposed model. Section 5 concludes the article.

	\section{Literature Survey}

	In modeling the portfolio credit risk, one requires the probability of the total number of defaults. The two categories of approach, namely the bottom-up approach and the top-down approach, have been proposed in literature. The top-down approach models the portfolio loss as a whole without making any reference to its constituents while the bottom-up approach models the default of the individual components which then are coupled in some way to obtain the portfolio loss. The top-down approach makes negligible reference to the portfolio components thus ignores the information, such as the marginal default probabilities, relevant particularly to individual components. Top-down models are investigated in various research articles. For instance, Errais et al. \cite{errais2010affine} studied the problem of portfolio credit risk by considering a family of point process whose intensity is assumed to be driven by affine jump-diffusion process. Arnsdorf and Halperin \cite{arnsdorf2009bslp} discussed the non-linear death process to model the arrivals of defaults in the portfolio. Davis and Lo \cite{davis2001modelling} modeled the arrival of defaults by a piecewise deterministic Markov process. Recently, Kim et al. \cite{kim2016hawkes} considered a Hawkes jump-diffusion process to model the conditional probability of defaults in the Eurozone. Giesecke \cite{gies2008} gives a comparison 
	between the two approaches for portfolio credit risk.

	The one-factor Gaussian copula model, based on the bottom-up approach, is a market standard model for modeling portfolio credit risk which provides an easily accessible tool to model the default correlation among individual constituents. However, the assumption of the Gaussian copula has several limitations and hence is not able to address the heavy tail dependence in the CDOs markets. Various copula models considering t-copula, Archimedean copulas are proposed to overcome the limitations of the Gaussian copula model (Li \cite{li1999default}, Frey \cite{frey2003dependent}, Schloegl and OKane \cite{schloegl2005note}, Laurent and Gregory \cite{laurent2005basket}). However, the copula approach is not appropriate because of two main reasons. Firstly, it is difficult to interpret the choice of copula and the parameters of the selected copula. Further, the dependence structure is imposed exogenously without justifying theoretically, and hence, the results are not robust with respect to the selected copula family and estimated parameters. Secondly, for the valuation of options on CDS portfolios or CDOs, a stochastic model is required for correlated credit spreads, but the standard copula approach does not address this issue.
	
	Given these limitations, a natural alternative to copulas is a multivariate version of the reduced form models. The class of dynamic bottom-up models for portfolio credit derivatives are reduced-form models that assumes the default time of a particular constituent as the first jump time of a Poisson process with stochastic intensity (also known as Cox process). Reduced form credit risk models were proposed by Jarrow and Turnbull \cite{jarrow1995pricing} and extended by Lando \cite{lando1998cox}, Sch$\ddot{o}$nbucher \cite{schonbucher1998term} and Duffie and Singleton \cite{duffie1999modeling}. This class of models, due to their mathematical tractability, has become very useful in the single-name credit markets. The first application to multi-name credit derivatives pricing was introduced by Duffie and Garleanu \cite{duffie2001risk} and later extended by Mortensen \cite{mortensen2005semi}. Allan Mortensen \cite{mortensen2005semi} presents a semi-analytical formula for pricing basket credit derivatives in an intensity-based model. Default intensities are assumed to be driven by correlated affine jump-diffusions. The intensity process of each component is further decomposed into systematic and idiosyncratic parts where each type of event is modeled by independent counting processes. Based on the same idea, Wu and Yang \cite{wu2010pricing} considered a mean-reverting diffusion process and derived a closed-form solution of tranche spreads in synthetic CDOs. Later, Wu and Yang \cite{wu2013valuation} extended their previous work by considering the default intensities to be driven by affine jump-diffusion processes involving L$\acute{e}$vy stable distributions. They derived an explicit formula for the expected loss of CDO tranches.
	
	The models by Allan Mortensen \cite{mortensen2005semi}, Wu and Yang \cite{wu2010pricing,wu2013valuation} are based on affine jump-diffusion processes, and the default is considered to be the first jump time of the jump process. However, due to the independent increments, the occurrence of a jump in Poisson jump-diffusion models does not stimulate future jumps. Hence, these models are incapable of capturing clustering of jumps both in time and across the markets. Further, A$\ddot{i}$t-Sahalia et al. \cite{ait2015} argued that ``what makes financial crisis take place is typically not the initial jump, but the amplification that occurs subsequently over hours or days, and the fact that other markets become affected as well". These empirical observations motivate us to develop a model that is competent to capture the feedback effect and the contagion effect. The natural choice to address the contagion and feedback phenomenon is a Hawkes process (Hawkes \cite{hawkes1971spectra}). Hawkes process is a self-exciting point process in the sense that the jumps in the process induce a jump in its intensity process. In other words, the intensity process is a function of the jumps of the point process itself. These processes now have widely been adopted to model the contagion effects in insurance finance and insurance, such as arrivals of trade, defaults in the credit market, etc.

	In order to account for the default clustering, Jarrow and Yu \cite{jarrow2001counterparty} proposed a contagion model where the default intensity of the counterparty is assumed to be a function of the default of the reference firm. More specifically, the probability of default of the counterparty is increased when the reference firm defaults. Errais et al. \cite{errais2010affine} extended their idea and proposed to use self-exciting Hawkes processes which can consistently address the clustering of defaults, as observed empirically on real data. Further, by including different parameters, they addressed firm-specific and systematic events. Dassios and Jang \cite{dassios2011dynamic} generalized the classical Hawkes process and introduced a dynamic contagion process which incorporates both self-exciting and external-exciting features. Their model was further extended by Dassios and Zhao \cite{dassios2011dynamic}, who introduced a diffusion component to capture the random behavior between any two jumps of the point process. {\color{blue} 
		In this article, we adopt the dynamic contagion process by Dassios and Zhao \cite{dassios2011dynamic} to propose a bottom-up framework to model the dynamics of joint defaults and obtain the prices of tranches of CDOs. Further, we present an alternative  parametric closed-form expression for the probability generating function (PGF) of the dynamic contagion process proposed in Dassios and Zhao \cite{dassios2011dynamic}.}

	\section{The Model Framework and Distribution of Total Number of Defaults} 
	
	We assume that the time horizon is finite, say $T > 0$. Further, assume that the uncertainty in the economy is governed by the filtered probability space $(\Omega,\mathcal{F},Q,\mathcal{F}_{t\in[0,T]})$, where $Q$ is a risk neutral probability measure. Consider a homogeneous collection of $N$ firms and assume the default time of $i$th firm, $i=1,2,\ldots,N$, is denoted by $\tau_i$ and $D_i(t)=I_{\{\tau_i\leq t\}}$ be the corresponding default indicator. Assume that $D(t)$ denotes the total number of defaults in the portfolio by time $t$. Therefore, we have
	\begin{equation}\label{totaldefaults}
		D(t):=\sum_{i=1}^ND_i(t).
	\end{equation}

	We assume that a number of hazardous events, for instance, a downgrade of credit rating or any unfavorable market news, relevant to any firm $i$ make it default. We model the arrival of these adverse events by a point process $\{\tilde{N}_i(t), t\geq 0\}$ such that the default probability of the $i$th firm is given by{\footnote{{\color{blue}Here, we follow the convention $0^0=1$.}}}
	\begin{equation}\label{defaultprob}
		P_i(T) = 1-E[(1-d_i)^{\tilde{N}_i(T)}],~~ d_i\in (0, 1].
	\end{equation}
	
	\begin{remark}
		Note that the proposed idea in Equation (\ref{defaultprob}) goes beyond the standard reduced-form models where the idea is to model the default intensity and default is defined as the first jump of the point process. The argument by A$\ddot{i}$t-Sahalia et al. \cite{ait2015} that ``what makes financial crisis take place is typically not the initial jump, but the amplification that occurs subsequently over hours or days, and the fact that other markets become affected as well". The authors observed that the default is caused by a number of events e.g., credit quality downgrade (or internal events), external events like major events in the economy or contagion effects. Therefore, one needs to model the arrival of bad events whose effect accumulates over time and cause the firm to default. Following Dassios and Zhao \cite{dassios2011dynamic}, our framework allows the $i$th company to have some resistance to survive through bad events, measured by $d_i$. Thus, $d_i$ can be considered as the measure of the capability{\footnote{\color{black}Unlike the standard reduced-form models where the first jump of the point process makes the firm default, we here assume that a jump in the point process can make the firm default with probability $d$. One can think of the probability $d$ as the resistance level of the firm against default. Furthermore, the firm survives upon the arrival of bad news with probability $1-d$. Therefore, the probability of survival is 
				\begin{equation*}
					P(\tau_i>T) = \sum_{n=0}^{\infty}P(\tau_i>T\mid \tilde{N}_i(T)=n)P(\tilde{N}_i(T)=n)=\sum_{n=0}^{\infty}(1-d_i)^nP(\tilde{N}_i(T)=n)=E[(1-d_i)^{\tilde{N}_i(T)}].
		\end{equation*}}}
		 of the $i$th firm to avoid default. For example, a higher credit rating corresponds to a lower value of $d_i$. 
	\end{remark}
	Assume that $\tilde{N}_i(t)$, {\color{blue}i.e., $\tilde{N}_i(t)=N_i(t) + \ell_i N(t)$}, is a combination of the common and idiosyncratic components $\{N(t),t\geq 0\}$ and $\{N_i(t),t\geq 0\}$ assumed to be independent of each other.  Therefore, the default probability of $i$th firm can be written as
	\begin{equation*}\label{defaultprob2}
		P_i(T) = 1-E[\theta_i^{N_i(T)}\tilde{d}_i^{N(T)}],
	\end{equation*}
	where we have
	\begin{equation}\label{measureresistance}
		\theta_i=(1-d_i),~\tilde{d}_i=(1-d_i)^{\ell_i}.
	\end{equation}
	Here, $\ell_i$ is a constant that represent the effect of the common component on the default probability of $i$th firm. Further, suppose that the stochastic processes $\{N(t),t\geq 0\}$ and $\{N_i(t),t\geq 0\}$ are are self exciting point processes with intensity processes $\{\lambda(t),t\geq 0\}$ and $\{\lambda_i(t),t\geq 0\}$ respectively{\footnote{{\color{black}We know that, in theory, the proposed intensity process can become negative with very small probability due to the term $\sigma dW(t)$. Nevertheless, the use of such process for default intensity is popular among both academicians and practitioners
				due to its analytical tractability (For reference, see \cite{kijima2000credit,pan2016reduced,pasricha2019pricing,tchuindjo2007pricing}). Further, one can take $\sigma=0$, if required, in order to avoid the possibility of negative values.}}}, given by
	\begin{equation}\label{idiointensity}
		d\lambda_i(t)=\delta_i(\eta_i-\lambda_i(t))dt+\sigma_idW_i(t)+d\left(\sum_{j=1}^{N_i(t)}Z^{(i)}_j\right),
	\end{equation}
	and
	\begin{equation}\label{commonintensity}
		d\lambda(t)=\delta(\eta-\lambda(t))dt+\sigma dW(t)+d\left(\sum_{j=1}^{N(t)}Z_j\right),
	\end{equation}
	where
	\begin{itemize}
		\item for any $i=1,2,\ldots,N$, the quantities $\lambda(0), \lambda_i(0),$ are the initial intensities at time $t=0$. Let $\lambda(0)=\lambda_0$ and $\lambda_i(0)=\lambda_{i,0}$ where $\lambda_0,\lambda_{i,0}$ are constants.
		\item for any $i=1,2,\ldots,N$, $\eta,\eta_i$ are the equilibrium level of the intensity processes.
		\item for any $i=1,2,\ldots,N$, $\delta,\delta_i$ are the constant mean-reversion speed of the intensity processes.
		\item for any $i=1,2,\ldots,N$, $\sigma,\sigma_i$ represent the volatility of the diffusion component.
		\item $\{Z_j\}_{j=1,2,\ldots}, \{Z^{(i)}_j\}_{j=1,2,\ldots}$, sequence of i.i.d. positive random variables representing the size of self-excited jumps. We assume these random variables have a common cumulative distribution function $G(z),~z>0$ and $G^{(i)}(z)$ respectively at the corresponding random times $\{T_j\}_{j=1,2,\ldots}, \{T^{(i)}_j\}_{j=1,2,\ldots}$. Further, we assume that $Z^{i}_j$ and $Z_j$ follows exponential distribution with parameter $\beta_i$ and $\beta$ respectively.
		\item $\{T_j\}_{j=1,2,\ldots}, \{T^{(i)}_j\}_{j=1,2,\ldots}$ are the arrival times of the processes $\{N(t),t\geq 0\}$  and $\{N_i(t),t\geq 0\}$ respectively.
		{\color{blue}
			\item $W:= \{W(t)|t \geq 0\}$ and $W_i:= \{W_i(t)|t \geq 0\}$ are the standard Brownian motions with respect to their natural filtration. 
			
			\item $\{Z_j\}_{j=1,2,\ldots}$, $\{Z^{(i)}_j\}_{j=1,2,\ldots}$, $\{T_j\}_{j=1,2,\ldots},\{T^{(i)}_j\}_{j=1,2,\ldots}$ and $\{W(t),t\geq 0\}, \{W^{(i)}(t),t\geq 0\}$ are independent of each other.}
	\end{itemize}

	\begin{remark}
		Note that the default intensity given in Equation (\ref{idiointensity}) is an appropriate choice because {\color{blue} the economic shocks specific to the $i$th company (short-lived endogenous shocks) and certain risk factors always persisting in the company are modeled respectively by the self-exciting jumps $N_i(t)$ and diffusion component $W_i(t)$. The corresponding market factors (exogenous risk shocks and ever-persisting risk factors) that affect all the companies are modeled respectively by the common self-exciting jump process $N(t)$ and diffusion component $W(t)$. Further, the default times of the constituents are independent given the common factor $N(t)$.}
	\end{remark}
	Therefore, the unconditional probability of total number of defaults i.e., $P(D(t)=j)$ is obtained as follows:
	\begin{equation}\label{joint}
		P(D(t)=j)=\sum_{n=0}^{\infty}P(D(t)=j\mid N(t)=n)P(N(t)=n).
	\end{equation}
	Note that the first factor on right hand side, representing the total number of defaults conditional on the common factor, follows a binomial distribution in a homogeneous portfolio pool{\footnote{\color{black} Note in Equation (\ref{binomial}) that $N$ is the number of firms in the portfolio (a constant) and $n$ is the number of jumps (can take values $0,1,2,\ldots$) in the common point process governing the default intensities of the firms.}}, i.e.,
	\begin{equation}\label{binomial}
		P(D(t)=j\mid N(t)=n)=\Mycomb[N]{j}(P_i(t,n))^j(1-P_i(t,n))^{N-j}
	\end{equation}
	where the marginal default probabilities given by using Equation (\ref{defaultprob}) as
	\begin{equation}\label{marginal}
		P_i(t,n)=P(\tau_i\leq t\mid N(t)=n)= 1-\tilde{d}_i^nE(\theta_i^{N_i(t)}),
	\end{equation}
	for each obligor $i$; $1\leq i\leq N$, where $n\in\{0,1,\ldots,\}$ is a constant and $\tilde{d}_i,\theta_i$ are given in Equation (\ref{measureresistance}). 
	
	On the other hand, we can compute $P(N(t)=n)$ from the probability generating function (PGF) of the process $N(t)$. In conclusion, we observe from Equations (\ref{joint}) and (\ref{marginal}) that it suffices to find the PGF of the point processes.

	\subsection{PGF of the Point Process $N(t)$}
	{{\color{blue} From the definition of $\lambda(t)$ and because of the assumption of the exponential kernel, i.e., the intensity decreases at the rate $\delta(\lambda-\eta)$, the process $\lambda(t)$ is a Markov process.} Further, on occurrence of a self-exciting jump, the process $N(t)$ increases by 1. Therefore, the process $(\lambda(t),N(t))$ is also a Markov process (refer \cite{dassios2011dynamic,dassios2017generalized,oakes1975markovian} for more details)}. This implies that the infinitesimal generator $\mathcal{A}$ of the process $(\lambda(t),N(t),t)$ applied to a function $f(\lambda,n,t)$ in $\Omega(\mathcal{A})$ is given by
	\begin{eqnarray}
		\nonumber \mathcal{A}f(\lambda,n,t)&=&\frac{\partial f}{\partial t}-\delta(\lambda-\eta)\frac{\partial f}{\partial \lambda}+\frac{1}{2}\sigma^2\frac{\partial^2 f}{\partial \lambda^{2}}\\
		\label{generator} &&+\lambda\int_0^{\infty}[f(\lambda+z,n+1,t)-f(\lambda,n,t)]dG(z),
	\end{eqnarray}
	Here, the domain of the generator $\mathcal{A}$, i.e, $\Omega(\mathcal{A})$ is such that the function $f(\lambda, n,t)$ is differentiable with respect to $\lambda,t$ for all $\lambda, n$ and $t$, and
	\begin{eqnarray*}
		|\int_0^{\infty}[f(\lambda + z, n+1,t) - f(\lambda, n,t)]dG(z)|<\infty.
	\end{eqnarray*}
	
	\begin{theorem}\label{mainlemma}
		Assume that jumps size $Z_j$ follows exponential distribution with parameter $\beta$ such that $\beta\delta>1$. Then, we have the following 
		\begin{eqnarray}\label{jointtransform}
			E(\theta^{(N(T))}e^{-v\lambda(T)}\mid\lambda_0)&=&e^{-B(0)\lambda(0)}e^{-(c(T)-c(0))},
		\end{eqnarray}
		where $0\leq\theta\leq 1$, $v\geq0$, are constants and for any $0\leq t\leq T$, the function $c(t)$ is given by
		\begin{eqnarray}
			\nonumber c(t)-c(0)&=&\eta\delta\int_0^tB(s)ds+\frac{1}{2}\sigma^2\int_0^tB^2(s)ds,
			\label{variable3b}
		\end{eqnarray}
		and $B(t)$ is given in parametric form as follows:
		\begin{eqnarray*}
			B(t(u))&=&\frac{(1+\delta\beta)}{\delta}u-\beta,
		\end{eqnarray*}
		where the time $t$ is a function of the parameter $u$ in the following form
		\begin{equation}\label{my_main_solution}
			t(u)=T-\frac{1}{\delta}\bigg(\frac{1}{2}\ln\bigg(\frac{u_0^2-u_0-\alpha_1}{u^2-u-\alpha_1}\bigg)+\frac{\tan^{-1}({\frac{2u_0-1}{\sqrt{-4\alpha_1-1}}})-\tan^{-1}({\frac{2u-1}{\sqrt{-4\alpha_1-1}}})}{\sqrt{-4\alpha_1-1}}\bigg),
		\end{equation}
		and {\color{blue} $u_0=\frac{(\beta+v)\alpha_1}{\alpha_2}$} with $\alpha_1=\frac{-\beta\theta\delta}{(1+\beta\delta)^2}$ and $\alpha_2=\frac{-\beta\theta}{1+\beta\delta}$.
	\end{theorem}
	\begin{proof}
		Assume that the function $f(\lambda,n,t)$ takes the following affine form
		\begin{equation}\label{affineform}
			f(\lambda,n,t)=A^n(t)e^{-B(t)\lambda}e^{c(t)},
		\end{equation}
		with the terminal condition $B(T)= v$. Substituting $f(\lambda,n,t)$ from Equation(\ref{affineform}) to Equation (\ref{generator}), yields
		\begin{eqnarray*}
			&&(-B'(t)\lambda+n\frac{A'(t)}{A(t)}+c'(t))f+B(t)\delta(\lambda-\eta)f+\frac{1}{2}\sigma^2B^2(t)f\\
			&& +\lambda f\int_0^{\infty}(A(t)e^{-B(t)z}-1)dG(z)=0.
		\end{eqnarray*}
		Since this equation is true for any possible choice of $\lambda$ and $n$, the coefficients of $\lambda$ and $n$ must be zero. Hence, it follows that
		\begin{eqnarray}
			\label{spde1} \frac{A'(t)}{A(t)}&=&0,~~~A(T)=\theta,\\
			\label{spde2} B'(t)-\delta B(t)-A(t)\frac{\beta}{\beta+B(t)}+1&=&0, ~~~ B(T)=v,\\
			\label{sspde3} c'(t)-\eta\delta+\frac{1}{2}\sigma^2B^2(t)&=&0.
		\end{eqnarray}
		From Equation (\ref{spde1}), we have $A(t)=\theta$ and Equation (\ref{sspde3}) is obtained as
		\begin{equation*}
			c(t)-c(0)=\eta\delta \int_0^t B(s)ds+\frac{1}{2}\sigma^2\int_0^tB^2(s)ds,
		\end{equation*}
		which is same as Equation (\ref{variable3b}). In order to get a closed form expression, we need to solve explicitly for $B(t)$, which can can be solved as follows:\\
		\noindent
		Define
		\begin{equation}\label{y_in_Bt}
			y(t)=(\beta+B(t))\exp\{\delta(T-t)\}.
		\end{equation}
		Then, we have
		\begin{equation*}
			\frac{dB(t)}{dt}=\exp\{-\delta(T-t)\}\frac{dy(t)}{dt}+\delta \exp\{-\delta(T-t)\}y(t),
		\end{equation*}
		and the terminal conditional on $B(t)$ gives {\color{blue}$y(T)=\beta+v$}. Therefore, Equation (\ref{spde2}) becomes
		\begin{equation}\label{eqn30}
			y(t)\frac{dy(t)}{dt}=-(\delta\beta+1)\exp\{\delta(T-t)\}y(t)+\beta\theta\exp\{2\delta(T-t)\}.
		\end{equation}
		The Equation (\ref{eqn30}) is an Abel's equation of second kind whose solution in a parametric form can be obtained in closed form (refer Section 1.3.3 in Polyanin and Zaitsev \cite{zaitsev2002handbook}). By aid of the substitution $x(t)=\int_t^T(1+\delta\beta)\exp\{\delta(T-t)\}dt$, we have
		\begin{equation}\label{substitution}
			x(t)=-\frac{(1+\delta\beta)}{\delta}(1-\exp\{\delta(T-t)\}),~~x(T)=0.
		\end{equation}
		After the substitution, we have
		\begin{equation}\label{eqn31}
			\frac{dy}{dt}=\frac{dy}{dx}\frac{dx}{dt}=-(1+\delta\beta)\exp\{\delta(T-t)\}\frac{dy}{dx}.
		\end{equation}
		From Equations (\ref{eqn30}) and (\ref{eqn31}), we can observe that we have
		\begin{equation*}
			y\frac{dy}{dx}-y=-\frac{\beta\theta}{1+\delta\beta}\exp\{\delta(T-t)\},
		\end{equation*}
		which further can be rewritten as
		\begin{eqnarray}
			\nonumber y\frac{dy}{dx}-y&=&-\frac{\beta\theta}{1+\delta\beta}\exp\{\delta(T-t)\}\\
			\nonumber&=&-\frac{\beta\theta}{1+\delta\beta}\bigg(\frac{\delta x}{1+\delta\beta}+1\bigg)\\
			\label{eqn34} &=&\alpha_1 x+\alpha_2,
		\end{eqnarray}
		where $\alpha_1=-\frac{\beta\theta\delta}{(1+\delta\beta)^2}$ and $\alpha_2=-\frac{\beta\theta}{1+\delta\beta}$.
		Using the parametric solution in Section 1.3.3 in Polyanin and Zaitsev \cite{zaitsev2002handbook}), we have
		\begin{eqnarray}
			\label{eqnfory}y(u)&=&Cu\exp\bigg(-\int_{-\infty}^u\frac{s}{s^2-s-\alpha_1}ds\bigg),\\
			\label{eqnforx} x(u)&=&C\exp\bigg(-\int_{-\infty}^u\frac{s}{s^2-s-\alpha_1}ds\bigg)-\frac{\alpha_2}{\alpha_1}.
		\end{eqnarray}
		We can observe from Equation (\ref{substitution}) that $x$ as a function of $u$ can be regarded as
		\begin{equation}\label{substitution2}
			x(u)=-\frac{(1+\delta\beta)}{\delta}(1-\exp\{\delta(T-t(u))\}),
		\end{equation}
		which gives
		\begin{equation}\label{t_in_u}
			t(u)=T-\frac{1}{\delta}\ln\bigg(1+\frac{\delta x(u)}{1+\delta\beta}\bigg).
		\end{equation}
		Further, observe that from Equations (\ref{eqnfory}) and (\ref{eqnforx}),
		\begin{equation*}
			y(u)=(x(u)+\frac{\alpha_2}{\alpha_1})u,
		\end{equation*}
		which further simplifies to
		\begin{equation*}
			y(u)=\frac{(1+\delta\beta)}{\delta}u\exp\{\delta(T-t(u))\}.
		\end{equation*}
		Thus, from Equation (\ref{y_in_Bt}), the parametric solution for $B(t)$ in terms of the parameter $u$ is given by
		\begin{eqnarray*}
			B(t(u))&=&y(u)\exp\{-\delta(T-t(u))\}-\beta\\
			&=&\frac{(1+\delta\beta)}{\delta}u-\beta.
		\end{eqnarray*}
		In order to complete the expression, we need to solve for the equation linking the time $t$ to the parameter $u$. We need to solve for $x(u)$ which in turn will give expression for $t(u)$ using Equation (\ref{t_in_u}).
		
		Let $u_0$ be the value of $u$ such that $t(u_0)=T$. In this case, from Equation (\ref{substitution2}), we have $x(u_0)=0$. This gives {\color{blue}$y(T)=y(t(u_0))=\beta+v$}. Therefore, substituting $u=u_0$ in Equation (\ref{eqnforx}) and (\ref{eqnfory}), we can solve for $C$ and $u_0$ which gives
		{\color{blue}
			\begin{eqnarray*}
				u_0&=&\frac{(\beta+v)\alpha_1}{\alpha_2},\\
				C&=&\frac{\alpha_2}{\alpha_1}\exp\bigg(\int_{-\infty}^{u_0}\frac{s}{s^2-s-\alpha_1}ds\bigg).
		\end{eqnarray*}}
		Hence, the expression of $x(u)$ is given by
		{\color{blue}
			\begin{equation*}
				x(u)=\frac{\alpha_2}{\alpha_1}\exp\bigg(\int_{u}^{\frac{(\beta+v)\alpha_1}{\alpha_2}}\frac{s}{s^2-s-\alpha_1}ds\bigg)-\frac{\alpha_2}{\alpha_1}.\\
		\end{equation*}}
		Therefore, $t(u)$ is given by (using Equation (\ref{t_in_u}))
		{\color{blue}
			\begin{equation*}
				t(u)=T-\frac{1}{\delta}\exp\bigg(\int_{u}^{\frac{(\beta+v)\alpha_1}{\alpha_2}}\frac{s}{s^2-s-\alpha_1}ds\bigg),
		\end{equation*}}
		{\color{black} which can be solved explicitly and is given in Equation (\ref{my_main_solution}). Since the process $(\lambda(t),N(t))$ is a Markov process, we know that the process
			\begin{equation}\label{martingale}
				M_t^f:=f((\lambda(t),N(t)))-f((\lambda(0),N(t)))-\int_0^t\mathcal{A}f((\lambda(s),N(s)))ds,   
			\end{equation}
			is a martingale (\cite{errais2010affine}). Taking $f((\lambda(t),N(t)))$ of the form  $\theta^{N(T)}e^{-B(T)\lambda(T)}e^{c(T)}$ and noting that $\mathcal{A}f((\lambda(t),N(t)))=0$, Equation (\ref{martingale}) yields that the process\\ $\theta^{N(T)}e^{-B(T)\lambda(T)}e^{c(T)}-e^{-B(0)\lambda(0)}e^{c(0)}$ is a martingale with zero mean, hence we have
			\begin{eqnarray*}
				E(\theta^{N(T)}e^{-v\lambda(T)}\mid\lambda_0)&=&e^{-B(0)\lambda(0)}e^{-(c(T)-c(0))}.
		\end{eqnarray*}}
	\end{proof}
	
	\begin{corollary}
		The PGF of $N(T)$ can be obtained from Theorem \ref{mainlemma} upon substituting $v=0$ and is given by
		\begin{eqnarray}\label{pgf}
			E(\theta^{N(T)}\mid\lambda_0)&=&e^{-B(0)\lambda(0)}e^{-(c(T)-c(0))},
		\end{eqnarray}
		where $B(0)$ and $c(T)-c(0)$ are given in Theorem \ref{mainlemma} with $v=0$.
	\end{corollary}
	
	\vspace{0.25cm}
	\begin{corollary}
		Similarly, the probability $P(N(T)= n\mid \lambda_0), n=0,1,2,\ldots$ can theoretically be obtained as follows
		\begin{equation}\label{probabilities}
			P(N(T)= n\mid \lambda_0)=\frac{\partial^n}{\partial\theta^n}E(\theta^{N(T)}\mid \lambda_0).
		\end{equation}
	\end{corollary}

	\section{Application to Valuation of Synthetic CDOs}

	In this section, we demonstrate the application of the proposed model in obtaining a closed-form expression for the tranche spread of a synthetic CDO.
	
	\subsection{The Cash Flow in a Synthetic CDOs and Valuation}
	
	We begin with a finite time horizon $T > 0$. Consider a portfolio of $N$ equally weighted single name CDS with default times of $i$th, $i=1,2,\ldots,N$, constituent of portfolio be denoted by $\tau_i$ and $D_i(t)=I_{\{\tau_i\leq t\}}$ be the default indicator. {\color{black} Further, we assume that recovery rate of each constituent is denoted by $w_i$ and is a constant $w_i\in[0,1),~i=1,2,\ldots,N$. Also, the nominal value for each constituent is assumed to be 1. Let $r$ denote the risk-free interest rate which is assumed to be a constant. Then, the loss of the portfolio $L(t)$, at any time $t$ is given by
		\begin{equation}\label{loss}
			L(t)=\frac{1}{N}\sum_{i=1}^N(1-w_i)D_i(t),
		\end{equation}
		which is a pure jump process}. A synthetic CDO with $M$ tranches is described by the points
	$$(0\% =)~0 = k_0 < k_1 < \ldots < k_{i-1} < k_{i} <\ldots < k_M = 1~(=100\%),$$
	where $[k_{i-1},k_{i}]$, $i=1,2,\ldots,M$ represent a tranche which covers losses in the portfolio between $k_{i-1}$ (called attachment point) and $k_{i}$ (called detachment point). More specifically, in a tranche $[k_{i-1},k_{i}]$, the protection seller is obliged to compensate the protection buyer for the losses occurring in the interval $[k_{i-1},k_{i}]$ in return to a premium that he receives periodically. The premium is proportional to the current outstanding on the value of the tranche. For a tranche $[k_{i-1},k_i]$, the payments made by the protection seller are at the corresponding default times and its expected value at time 0 is denoted by $V_i(T)$. On the other hand, the expected value of premium paid by the protection buyer is denoted by $W_i(T)$. The accumulated loss $L_{i}(t)$ of $i$th tranche at time $t$ is
	\begin{equation}\label{protectionleg1}
		L_{i}(t)=(L(t)-k_{i-1})I_{\{L(t)\in[k_{i-1},k_{i}]\}}+(k_{i}-k_{i-1})I_{\{L(t)\geq k_{i}\}},
	\end{equation}
	where $I_{A}$ is the indicator function of the set $A$. {\color{black} We can note that $L_i(t)$ is also a pure jump process such that the default payments are nothing but increments of the process $L_i(t)$. The expected value, $V_i(T)$, for $i$th tranche is given by
		\begin{equation}\label{protectionleg}
			V_{i}(T)=E\bigg(\int_0^Te^{-rt}dL_i(t)\bigg)=e^{-rT}E(L_{i}(T))+r\int_0^Te^{-rt}E(L_{i}(t))dt.
		\end{equation}
		where the second equality follows by applying integration by parts for Lebesgue–Stieltjes measures along with Fubini–Tonelli theorem}.
	
	Assume that the $n$ premium payment dates are $0 < t_1 < t_2 <\ldots< t_n = T$, then the expected value, $W_i(T)$, for $i$th tranche is given by
	\begin{equation}\label{premiumleg}
		W_{i}(T) = c_{i}(T)\sum_{j=1}^ne^{-rt_j}(\Delta k_{i}-E(L_{i}(t_j)))\Delta t_j,
	\end{equation}
	where we have assumed that there is no accrual payments at default, $\Delta t_j=(t_j-t_{j-1})$ denotes the time duration between the payments (here measured in fractions of a year), $\Delta k_{i}=k_{i}-k_{i-1}$ is the nominal size of the $i$th tranche (as a fraction of the total nominal value of the portfolio), and $c_{i}(T)$ is called the spread of the tranche.  $c_i(T)$ can be obtained by equating the value of two expected premiums, i.e., $V_{i}(T) = W_{i}(T)$. The tranche spreads $c_{i}(T)$, for $i=1,2,\ldots,M$, are functions of the time horizon $T$.
	
	Now, in order to obtain the spread of the $i$th tranche i.e., $c_i(T)$, we need to compute the expected loss of the $i$th tranche i.e., $E(L_{i}(t))$ which is given by
	(from Equation (\ref{protectionleg1})),
	\begin{equation}\label{trancheloss}
		E(L_{i}(t))=(k_{i}-k_{i-1})P(L(t)>k_{i})+\int_{k_{i-1}}^{k_{i}}(x-k_{i-1})dF_{L(t)}(x),
	\end{equation}
	where $F_{L(t)}(x)$ denotes the cumulative distribution function of the loss of the portfolio by time $t$, i.e.,
	$$F_{L(t)}(x) := P(L(t)\leq x).$$ 
	The cumulative loss process $L(t)$ being a pure jump process, Equation (\ref{trancheloss}) can further be written as 
	\begin{equation}\label{trancheloss2}
		E(L_{i}(t))=(k_{i}-k_{i-1})P(L(t)>k_{i})+\sum_{j=k_{i-1}}^{k_{i}}(j-k_{i-1})P(L(t)=j).
	\end{equation}
	Note that the index $j$ in summation corresponds only to possible values of the portfolio loss $L(t)$ at time $t$ that are in the interval $[k_{i-1},k_i]$. From Equation (\ref{trancheloss2}), it can be observed that we need to compute the distribution of the portfolio loss in order to find the tranche loss $E[L_{i}(t)]$. In order to find $P(L(t)=k)$ for some $k$, a possible value of the portfolio loss at time $t$, we proceed by finding the PGF of $L(t)$. One can then use the PGF to find the required probabilities of $L(t)$.
	\begin{eqnarray*}
		E(u^{L(t)})&=&\sum_{n=0}^{\infty}E\bigg(u^{\sum_{i=0}^N\frac{M_i}{N}D_i(t)}\mid N(t)=n\bigg)P(N(t)=n)
	\end{eqnarray*}
	Given $N(t)$, the default of each constituent is independent, which yields
	\begin{eqnarray*}
		E(u^{L(t)})&=&\sum_{n=0}^{\infty}\bigg(\prod_{i=1}^NE\bigg( u^{\frac{M_i}{N}D_i(t)}\mid N(t)=n\bigg)\bigg)P(N(t)=n)\\
		&=&\sum_{n=0}^{\infty}\bigg(\prod_{i=1}^N\bigg(1-P_i(t,n)+P_i(t,n) u^{\frac{M_i}{N}}\bigg)\bigg)P(N(t)=n).
	\end{eqnarray*}
	{\color{blue}
		The procedure to calculate the unknowns in the above equation, i.e., $P_i(t,n)$ and $P(N(t)=n)$ are given by Equations (\ref{marginal}) and (\ref{probabilities}) respectively. Note that the calculation of $P_i(t,n)$ requires PGF of $N(t)$ which has already been derived in Equation (\ref{pgf}).}

	\noindent
	Next, we discuss the case when the recovery rates are uniform, $w_i=w,~~i=1,2,\ldots,N$. 
	
	\begin{theorem}
		Assume that the recovery rates are uniform, i.e.,  $w_i=w,~~i=1,2,\ldots,N$. The expected tranche losses of a synthetic CDO are given by
		\begin{eqnarray*}
			E(L_{i}(t))&=&(k_{i}-k_{i-1})-k_{i}\sum_{j=0}^{[\frac{Nk_{i}}{1-w}]}\sum_{n=0}^{\infty}P(D(t)=j\mid N(t)=n)P(N(t)=n)\\
			&& +k_{i-1}\sum_{j=0}^{[\frac{Nk_{i-1}}{1-w}]}\sum_{n=0}^{\infty}P(D(t)=j\mid N(t)=n)P(N(t)=n)\\
			&& +\sum_{j=[\frac{Nk_{i-1}}{1-w}]}^{[\frac{Nk_{i}}{1-w}]}j\sum_{n=0}^{\infty}P(D(t)=j\mid N(t)=n)P(N(t)=n).
		\end{eqnarray*}
	\end{theorem}
	\begin{proof}
		{\color{blue} From Equations (\ref{loss}) and (\ref{totaldefaults}), we have $L(t)=(1-w)\frac{D(t)}{N}$. Therefore, for $x\in [0,1-w]$, it holds that
			\begin{equation}
				F_{L(t)}(x)= P(L(t)\leq x)=P\left(D(t)\leq \frac{xN}{1-w}\right)=\sum_{i=0}^{[\frac{xN}{1-w}]}P(D(t)=i).
			\end{equation}
			where $P(D(t)=i)$ is given by Equation (\ref{joint}). Therefore, the expected tranche loss can be obtained as follows: } 
		\begin{eqnarray*}
			E(L_{i}(t))&=&(k_{i}-k_{i-1})P(L(t)\geq k_{i})+\int_{k_{i-1}}^{k_{i}}(x-k_{i-1})dF_{L(t)}(x)\\
			&=& (k_{i}-k_{i-1})(1-P(L(t)\leq k_{i})+\int_{k_{i-1}}^{k_{i}}xdF_{L(t)}(x)-k_{i-1}(F_{L(t)}(k_{i})-F_{L(t)}(k_{i-1})\\
			&=&k_{i}-k_{i-1}-k_{i}F_{L(t)}(k_{i})+k_{i-1}F_{L(t)}(k_{i-1})+\int_{k_{i-1}}^{k_{i}}xdF_{L(t)}(x).
		\end{eqnarray*}
		
		Now, $\int_{k_{i-1}}^{k_{i}}xdF_{L(t)}(x)=E(L(t)\cap [k_{i-1},k_{i}])$. Hence, by definition
		\begin{equation*}
			\int_{k_{i-1}}^{k_{i}}(x)dF_{L(t)}(x)=\sum_{j=[\frac{Nk_{i-1}}{1-w}]}^{[\frac{Nk_{i}}{1-w}]}jP(D(t)=j),
		\end{equation*}
		since $L(t)\in[k_{i-1},k_i]$ is equivalent to $D(t)\in[[\frac{Nk_{i-1}}{1-w}],[\frac{Nk_{i}}{1-w}]]$. Hence, the result follows.
	\end{proof}
	
	\subsection{Numerical Illustrations}
	This section provides sensitivity analysis on the spread of the tranches and comparison with models proposed in the literature, to demonstrate the impact of jump clustering and various parameters governing the contagion. In what follows, we consider three tranches namely $0\%-7\%, 7\%-12\%$ and $100\%$ and the number of constituents in the portfolio are $M=50$. Table \ref{basecase} lists the values  of the parameters considered in the base case. Further, for numerical experiments, the value of all the parameters except one under study is kept unchanged. Finally, we assume that the premium payment dates are quarterly and hence the parameters are assumed to be unit per quarter. For instance, $\lambda_0=1.5$ is per 3 months period.
	
	\begin{table}[h!]
		\begin{center}
			\caption{Values of the Parameters in the Base Case}\label{basecase}
			\begin{tabular}{cc|cc}
				\hline
				Parameters& Values & Parameters& Values\\
				\hline
				$N$ & 50 & $r$ & 0.03 \\
				$\sigma_{i}$ & 0.4 & $\sigma$ & 0.4 \\
				$\lambda_{i,0}$&1.5 & $\lambda_{0}$& 1.5\\
				$\eta_{i}$ & 1.5 & $\eta$ & 1.5 \\
				$\beta_{i_y}$&1.5 & $\beta_y$& 1.5\\
				$\delta_i$& 2 &$\delta$&2\\
				$w$&0.4&$\theta_i$&0.97\\
				$\ell_i$&0.5&&\\
				\hline
			\end{tabular}
		\end{center}
	\end{table}
	In Table \ref{maturityvary}, the spread of each tranche is given for different maturity times for the Poisson process (constant intensity), affine jump-diffusion model with no contagion effects and for the proposed model. From the Table \ref{basecase}, we observe that the spread of each tranche is an increasing function of time for all the models. Also, we compare the spreads obtained from the proposed model to the spreads obtained from a Poisson process model (constant intensity of default) given in Table \ref{maturityvarypoisson} and spreads obtained from a Cox process with affine jump diffusion intensity with no self-exciting component given in Table \ref{maturityvarynoself}. For the comparison purpose, the default intensities in case of Poisson process model and Cox process model are selected so that the expected number of jumps per unit time are same for all the models. It is observed from Table \ref{maturityvary} and Table \ref{maturityvarypoisson}, that there is flat spread term structure whereas in case of the proposed model and tranche spread is much higher in case of Poisson process model. Possible reason for this behavior is due to the assumption that in Poisson process model, first jump causes the default of the firm whereas in the proposed model, each firm has some resistance to the bad events. Also, since relatively lesser number of defaults wipe out the lower tranches causing flat term structure of the spread for these tranches.
	
	\begin{table}[h!]
		\begin{center}
			\caption{Spread of tranches (in 100 bps) with maturity time $T$ in the Base Case (Proposed Model)}\label{maturityvary}
			\begin{tabular}{|c|c|c|c|c|}
				\hline
				& $T=3$ & $T=4$ & $T=5$ & $T=6$ \\
				\hline
				Tranche 1 & 0.5343 & 0.5584 & 0.5654 & 0.5689 \\
				\hline
				Tranche 2 & 0.0421 & 0.0921 & 0.1359 & 0.1660 \\
				\hline
				Tranche 3 & 0.00006 & 0.0001 & 0.0005 & 0.0012 \\
				\hline
			\end{tabular}
		\end{center}
	\end{table}

	\begin{table}[h!]
		\begin{center}
			\caption{Spread of tranches (in 100 bps) with maturity time $T$ in the Base Case (Poisson process model)}\label{maturityvarypoisson}
			\begin{tabular}{|c|c|c|c|c|}
				\hline
				& $T=3$ & $T=4$ & $T=5$ & $T=6$ \\
				\hline
				Tranche 1 & 3.9439 & 3.9439 & 3.9439 & 3.9439 \\
				\hline
				Tranche 2 & 2.6154 & 2.6154 & 2.6154 & 2.6154 \\
				\hline
				Tranche 3 & 0.0779& 0.0614 & 0.0506 & 0.0432 \\
				\hline
			\end{tabular}
		\end{center}
	\end{table}
	
	\begin{table}[h!]
		\begin{center}
			\caption{Spread of tranches (in 100 bps) with maturity time $T$ in the Base Case (Affine jump diffusion without self-exciting jumps)}\label{maturityvarynoself}
			\begin{tabular}{|c|c|c|c|c|}
				\hline
				& $T=3$ & $T=4$ & $T=5$ & $T=6$ \\
				\hline
				Tranche 1 & 1.9368 & 1.9368&1.9368&1.9368 \\
				\hline
				Tranche 2 & 0.9758 & 0.9758&0.9758&0.9758 \\
				\hline
				Tranche 3 & 0.0357& 0.0357 & 0.0334 & 0.0306  \\
				\hline
			\end{tabular}
		\end{center}
	\end{table}

	Figures \ref{TP1lambda}, \ref{TP2lambda} and \ref{TP3lambda} present the variation of tranche spread for each tranche against the change in the initial intensity of the common process i.e., $\lambda_0$. From Figures \ref{TP1lambda}, \ref{TP2lambda} and \ref{TP3lambda}, it is observed that a increase in the value of $\lambda_0$ increases the spread of each tranche and the effect becomes more significant with the increase in maturity. Since jump clustering increases with increase in $\lambda_0$, hence jump clustering have a positive effect, i.e., increases the tranche spread.
	
	Similarly, Figures \ref{TP1betay}, \ref{TP2betay} and \ref{TP3betay} presents the variation of tranche spread for each tranche against the change in the $\beta_y$ (and hence the mean of the jump size) of the common process. From Figures \ref{TP1betay}, \ref{TP2betay} and \ref{TP3betay}, we observe that the tranche spread increases for each tranche with the decrease in the value of parameter of the jump size distribution i.e., $\beta_y$. As we know, larger the jump size i.e., smaller the value of $\beta_y$ because expected jump size is reciprocal of $\beta_y$, implies more clustering of jumps. Hence, we conclude that the tranche spread for each tranche increases with the clustering of jumps.
	
	In conclusion, higher the value of $\lambda_0$ and smaller the value of $\beta_y$ (i.e., self exciting jumps of larger size), the more frequent is the clustering of events, and hence higher is the spread and steeper is the term structure of spread.

	Figures \ref{TP1a}, \ref{TP2a} and \ref{TP3a} give the behavior of tranche spread for each tranche against the change in the mean reverting level of the intensity process governing the common factor i.e., $\eta$. We observe from Figures \ref{TP1a}, \ref{TP2a} and \ref{TP3a} that when mean reverting level $\eta$ is less than the initial intensity $\lambda_0$, the spread of each tranche is robust, i.e., does not change significantly. However, as $\eta$ increases beyond the initial intensity $\lambda_0$, spread of tranche increases significantly with increase in $\eta$.

	It can be observed from Figures \ref{TP1w}, \ref{TP2w} and \ref{TP3w}, as the recovery rate increases, the accumulated loss due to observed defaults decreases and hence the tranche spread decreases as expected. Further, Figures \ref{TP1sigma}, \ref{TP2sigma} and \ref{TP3sigma} show the robustness of the tranche spread with respect to the volatility of the diffusion term governing the default intensity. Also, it can be observed that the sensitivity of the spread of tranches increases with the decreases with the increase in seniority of the tranches. In other words, senior tranche spread are more robust with respect to volatility of diffusion term.
	
	\section{Conclusion}

	In this article, we proposed a framework to model the joint default of firms by taking into account the clustering phenomenon observed empirically in the occurrence of default events. The point process governing the default arrivals is assumed to be composed of idiosyncratic and systematic components. The default intensity of each process is further being modeled by a self-exciting Hawkes process with a diffusion component to address the stochastic behavior between the two jumps in the intensity process. One of the distinguishing features of the proposed framework is that it allows firms to have resistance against the adverse events and thus doesn't assume the default to be caused by a first adverse event. We presented a parametric closed-form solution for the joint default probabilities represented in terms of the probability generating function of a dynamic contagion process. Finally, we studied the application in pricing the synthetic CDO in a bottom-up approach. The proposed framework is capable of addressing the chaotic dynamics and high complexity in the pricing of CDOs. Although, we have shown by numerical illustrations that the proposed model is appropriate to address the clustering of defaults observed in the default events. One of the possible future work is to formulate how to calibrate the model parameters using real data and perform the parameter estimation.

	\textbf{Acknowledgements}
		The first author is thankful to Council of Scientific and Industrial Research (CSIR), India (award number 09/086 (1207)/2014-EMR-I) for the financial grant. Authors are thankful to the editor and the anonymous reviewer for their valuable suggestions and comments which helped improve the paper to a great extent.
	
	\bibliographystyle{tfs}       
	\bibliography{mybibbottomup}   

\begin{thebibliography}{10}
\providecommand{\MR}{\relax\unskip\space MR }
\providecommand{\url}[1]{\normalfont{#1}}
\providecommand{\urlprefix}{Available at }

\bibitem{ait2015}
Y. A{\"\i}t-Sahalia, J. Cacho-Diaz, and R.J. Laeven, \emph{Modeling financial
  contagion using mutually exciting jump processes}, Journal of Financial
  Economics 117 (2015), pp. 585--606.

\bibitem{arnsdorf2009bslp}
M. Arnsdorf and I. Halperin, \emph{Bslp: Markovian bivariate spread-loss model
  for portfolio credit derivatives}, Journal of Computataional Finance 12
  (2008), pp. 77--107.

\bibitem{dassios2011dynamic}
A. Dassios and H. Zhao, \emph{A dynamic contagion process}, Advances in Applied
  Probability 43 (2011), pp. 814--846.

\bibitem{dassios2017generalized}
A. Dassios and H. Zhao, \emph{A generalized contagion process with an
  application to credit risk}, International Journal of Theoretical and Applied
  Finance 20 (2017), p. 1750003.

\bibitem{davis2001modelling}
M. Davis and V. Lo, \emph{Infectious defaults}, Quantitative Finance 1 (2001),
  pp. 382--387.

\bibitem{duffie2001risk}
D. Duffie and N. Garleanu, \emph{Risk and valuation of collateralized debt
  obligations}, Financial Analysts Journal 57 (2001), pp. 41--59.

\bibitem{duffie1999modeling}
D. Duffie and K.J. Singleton, \emph{Modeling term structures of defaultable
  bonds}, The Review of Financial Studies 12 (1999), pp. 687--720.

\bibitem{errais2010affine}
E. Errais, K. Giesecke, and L.R. Goldberg, \emph{Affine point processes and
  portfolio credit risk}, SIAM Journal on Financial Mathematics 1 (2010), pp.
  642--665.

\bibitem{frey2003dependent}
R. Frey and A.J. McNeil, \emph{Dependent defaults in models of portfolio credit
  risk}, Journal of Risk 6 (2003), pp. 59--92.

\bibitem{gies2008}
K. Giesecke, \emph{Portfolio credit risk: top-down vs bottom-up}, in
  \emph{Frontiers in Quantitative Finance: Volatility and Credit Risk
  Modeling}, R. Cont, ed., chap.~10, Wiley,  2008, pp. 251--267.

\bibitem{hawkes1971spectra}
A.G. Hawkes, \emph{Spectra of some self-exciting and mutually exciting point
  processes}, Biometrika 58 (1971), pp. 83--90.

\bibitem{jarrow1995pricing}
R.A. Jarrow and S.M. Turnbull, \emph{Pricing derivatives on financial
  securities subject to credit risk}, The Journal of Finance 50 (1995), pp.
  53--85.

\bibitem{jarrow2001counterparty}
R.A. Jarrow and F. Yu, \emph{Counterparty risk and the pricing of defaultable
  securities}, the Journal of Finance 56 (2001), pp. 1765--1799.

\bibitem{kijima2000credit}
M. Kijima and Y. Muromachi, \emph{Credit events and the valuation of credit
  derivatives of basket type}, Review of Derivatives Research 4 (2000), pp.
  55--79.

\bibitem{kim2016hawkes}
J. Kim, Y.J. Park, and D. Ryu, \emph{Hawkes-diffusion process and the
  conditional probability of defaults in the {E}urozone}, Physica A:
  Statistical Mechanics and its Applications 449 (2016), pp. 301--310.

\bibitem{lando1998cox}
D. Lando, \emph{On cox processes and credit risky securities}, Review of
  Derivatives Research 2 (1998), pp. 99--120.

\bibitem{laurent2005basket}
J.P. Laurent and J. Gregory, \emph{Basket default swaps, {CDO}s and factor
  copulas}, The Journal of Risk 7 (2005), p.~1.

\bibitem{li1999default}
D.X. Li, \emph{On default correlation: A copula function approach}, Journal of
  Fixed Income 9 (2000), pp. 43--54.

\bibitem{mortensen2005semi}
A. Mortensen, \emph{Semi-analytical valuation of basket credit derivatives in
  intensity-based models}, The Journal of Derivatives 13 (2006), pp. 8--26.

\bibitem{oakes1975markovian}
D. Oakes, \emph{The {M}arkovian self-exciting process}, Journal of Applied
  Probability  (1975), pp. 69--77.

\bibitem{pan2016reduced}
J. Pan and Q. Xiao, \emph{A reduced-form model for pricing defaultable bonds
  and credit default swaps with stochastic recovery}, Applied Stochastic Models
  in Business and Industry 32 (2016), pp. 725--739.

\bibitem{pasricha2019pricing}
P. Pasricha and A. Goel, \emph{Pricing vulnerable power exchange options in an
  intensity based framework}, Journal of Computational and Applied Mathematics
  355 (2019), pp. 106--115.

\bibitem{pasricha2019markov}
P. Pasricha and D. Selvamuthu, \emph{A {M}arkov modulated dynamic contagion
  process with application to credit risk}, Journal of Statistical Physics 175
  (2019), pp. 495--511.

\bibitem{schloegl2005note}
L. Schloegl and D. O'kane,
 \emph{A note on the large homogeneous portfolio
  approximation with the student-t copula}, Finance and Stochastics 9 (2005),
  pp. 577--584.

\bibitem{schonbucher1998term}
P.J. Sch{\"o}nbucher, \emph{Term structure modelling of defaultable bonds},
  Review of Derivatives Research 2 (1998), pp. 161--192.

\bibitem{tchuindjo2007pricing}
L. Tchuindjo, \emph{Pricing of multi-defaultable bonds with a
  two-correlated-factor hull--white model}, Applied Mathematical Finance 14
  (2007), pp. 19--39.

\bibitem{wu2010pricing}
J.L. Wu and W. Yang, \emph{Pricing {CDO}s tranches in an intensity based model
  with the mean reversion approach}, Mathematical and Computer Modelling 52
  (2010), pp. 814--825.

\bibitem{wu2013valuation}
J.L. Wu and W. Yang, \emph{Valuation of synthetic {CDO}s with affine
  jump-diffusion processes involving {L}{\'e}vy stable distributions},
  Mathematical and Computer Modelling 57 (2013), pp. 570--583.

\bibitem{zaitsev2002handbook}
V.F. Zaitsev and A.D. Polyanin, \emph{Handbook of exact solutions for ordinary
  differential equations}, Chapman and Hall/CRC, 2002.

\end{thebibliography}

	\begin{figure}[ht!]
		\begin{center}
			\subfigure[Spread of Tranche 1 (bps) against $\lambda_0$]{%
				\includegraphics[scale=0.4500]{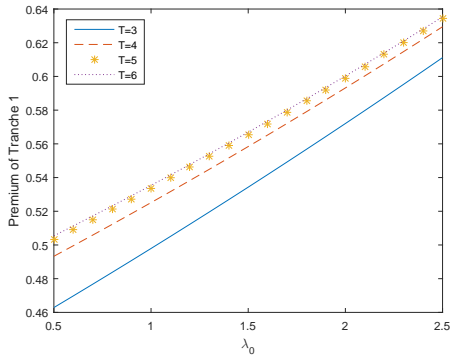}
				\label{TP1lambda}}
			\quad
			\subfigure[Spread of Tranche 1 (bps) against $\beta_y$]{%
				\includegraphics[scale=0.4500]{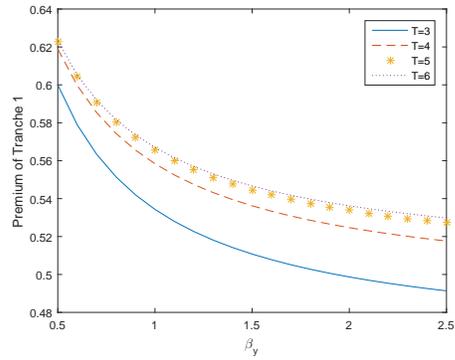}
				\label{TP1betay}}
			\quad
			\subfigure[Spread of Tranche 1 (bps) against $\eta$]{%
				\includegraphics[scale=0.45000]{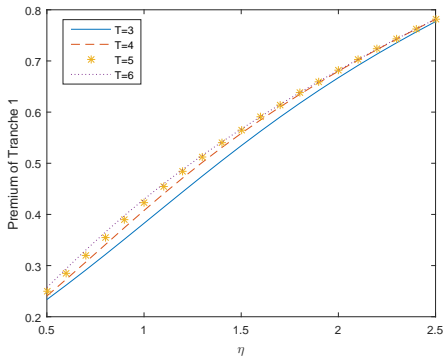}
				\label{TP1a}}
			\quad
			\subfigure[Spread of Tranche 1 (bps) against $w$]{%
				\includegraphics[scale=0.4500]{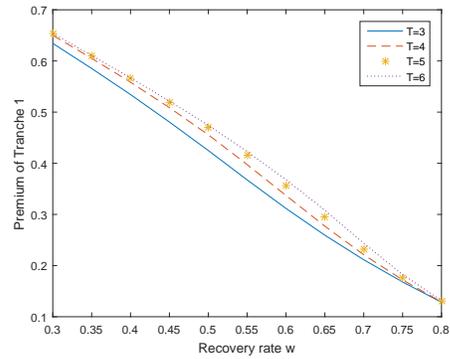}
				\label{TP1w}}
			\quad
			\subfigure[Spread of Tranche 1 (bps) against $\sigma$]{%
				\includegraphics[scale=0.450]{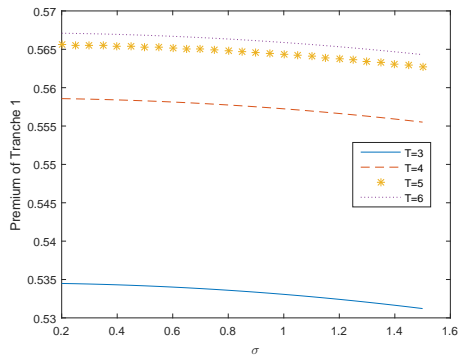}
				\label{TP1sigma}}
			\caption{Variation of Spread of Tranche 1 against different parameters governing the default of firms for different maturities for $T=3,4,5,6$ years}\label{TP1}
		\end{center}
	\end{figure}

	\begin{figure}[ht!]
		\begin{center}
			\subfigure[Spread of Tranche 2 (bps) against $\lambda_0$]{%
				\includegraphics[scale=0.4500]{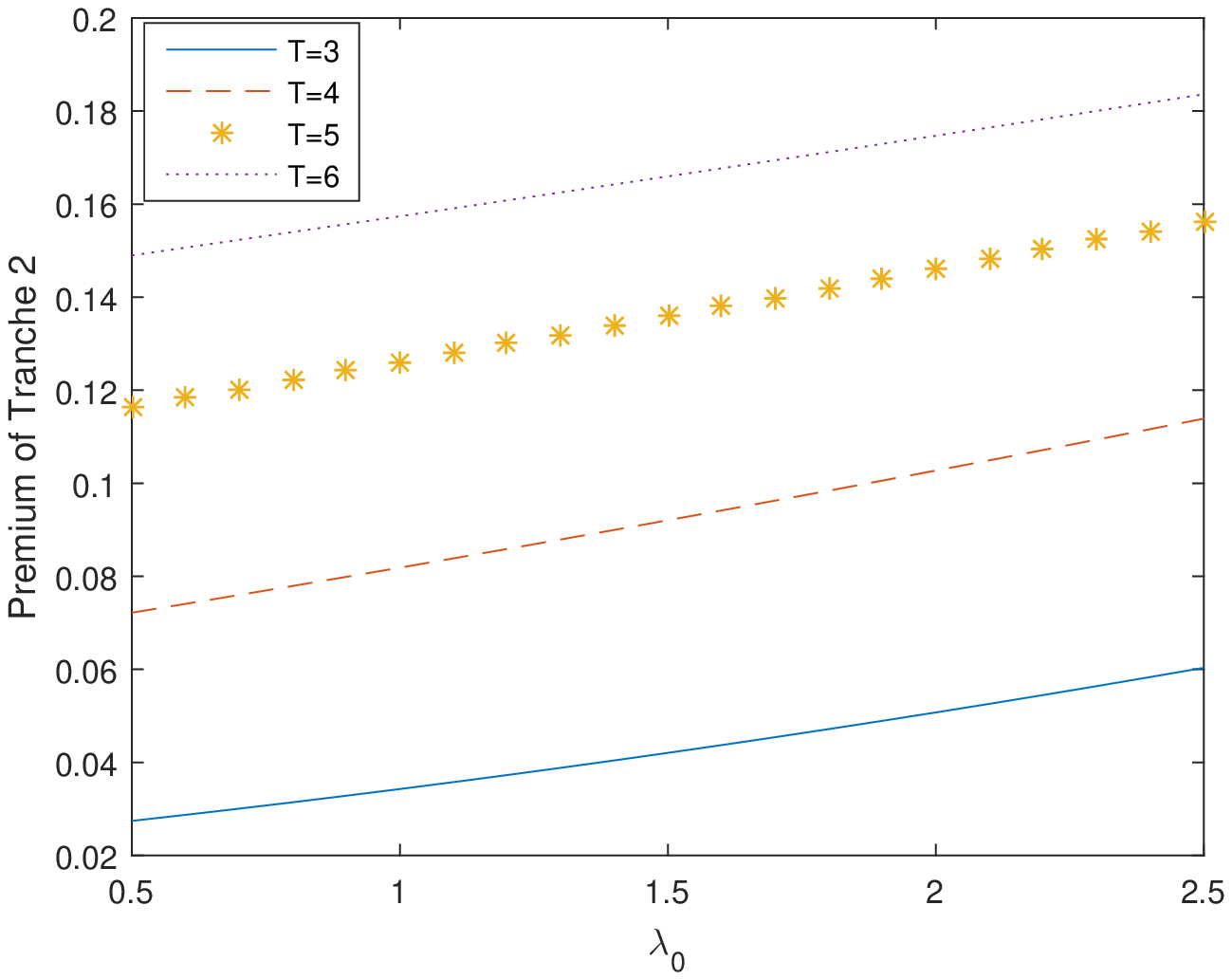}
				\label{TP2lambda}}
			\quad
			\subfigure[Spread of Tranche 2 (bps) against $\beta_y$]{%
				\includegraphics[scale=0.4500]{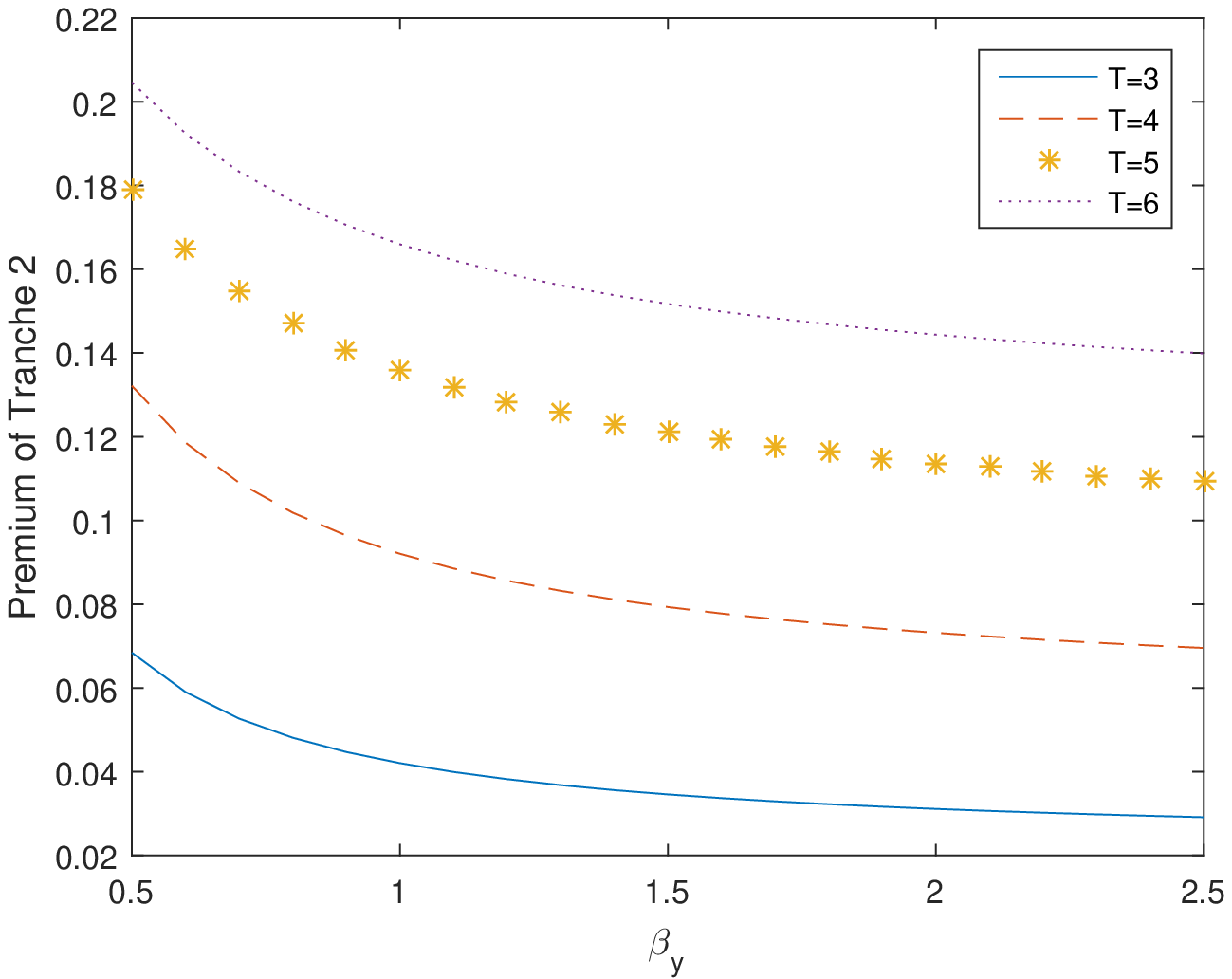}
				\label{TP2betay}}
			\quad
			\subfigure[Spread of Tranche 2 (bps) against $\eta$]{%
				\includegraphics[scale=0.4500]{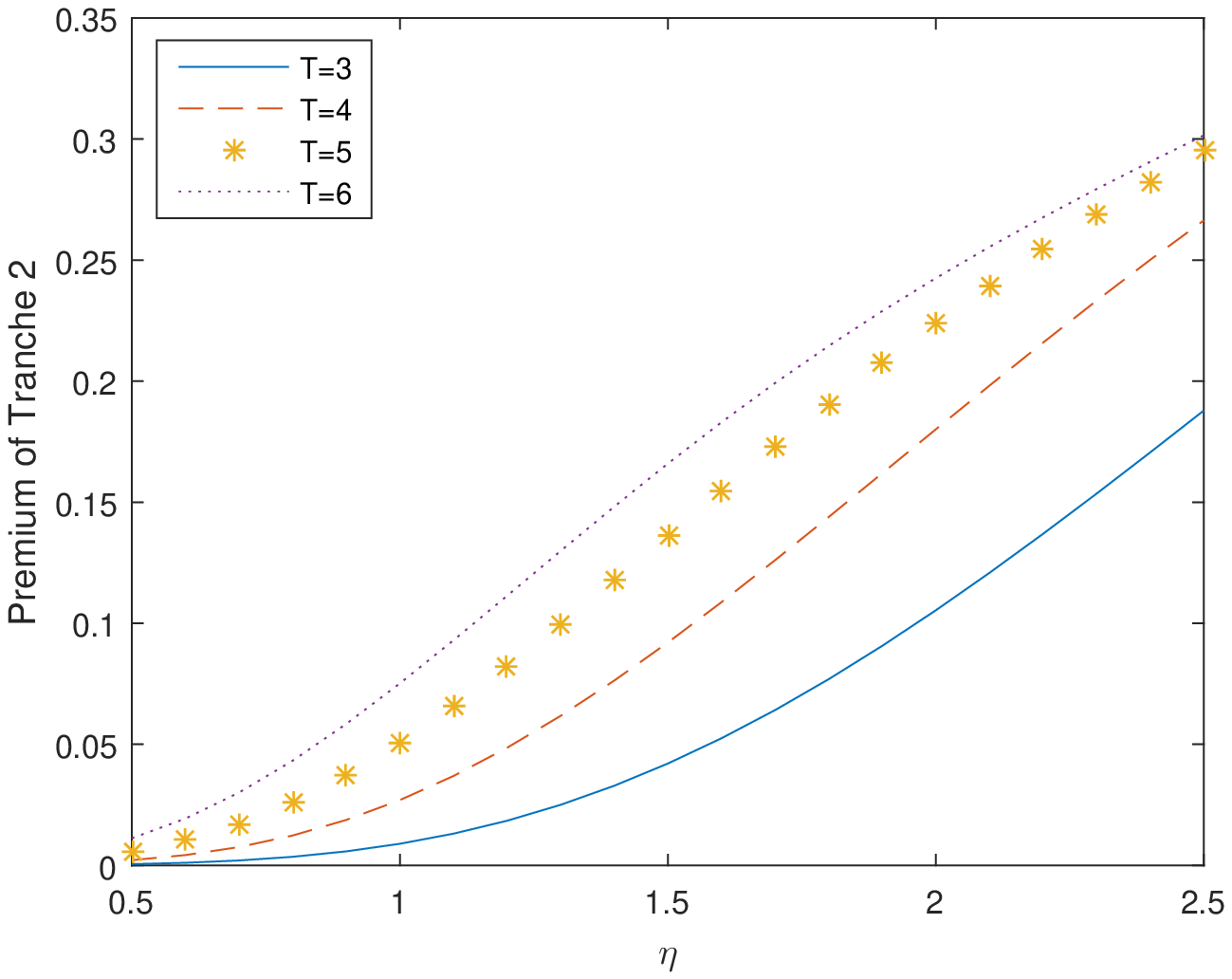}
				\label{TP2a}}
			\quad
			\subfigure[Spread of Tranche 2 (bps) against $w$]{%
				\includegraphics[scale=0.4500]{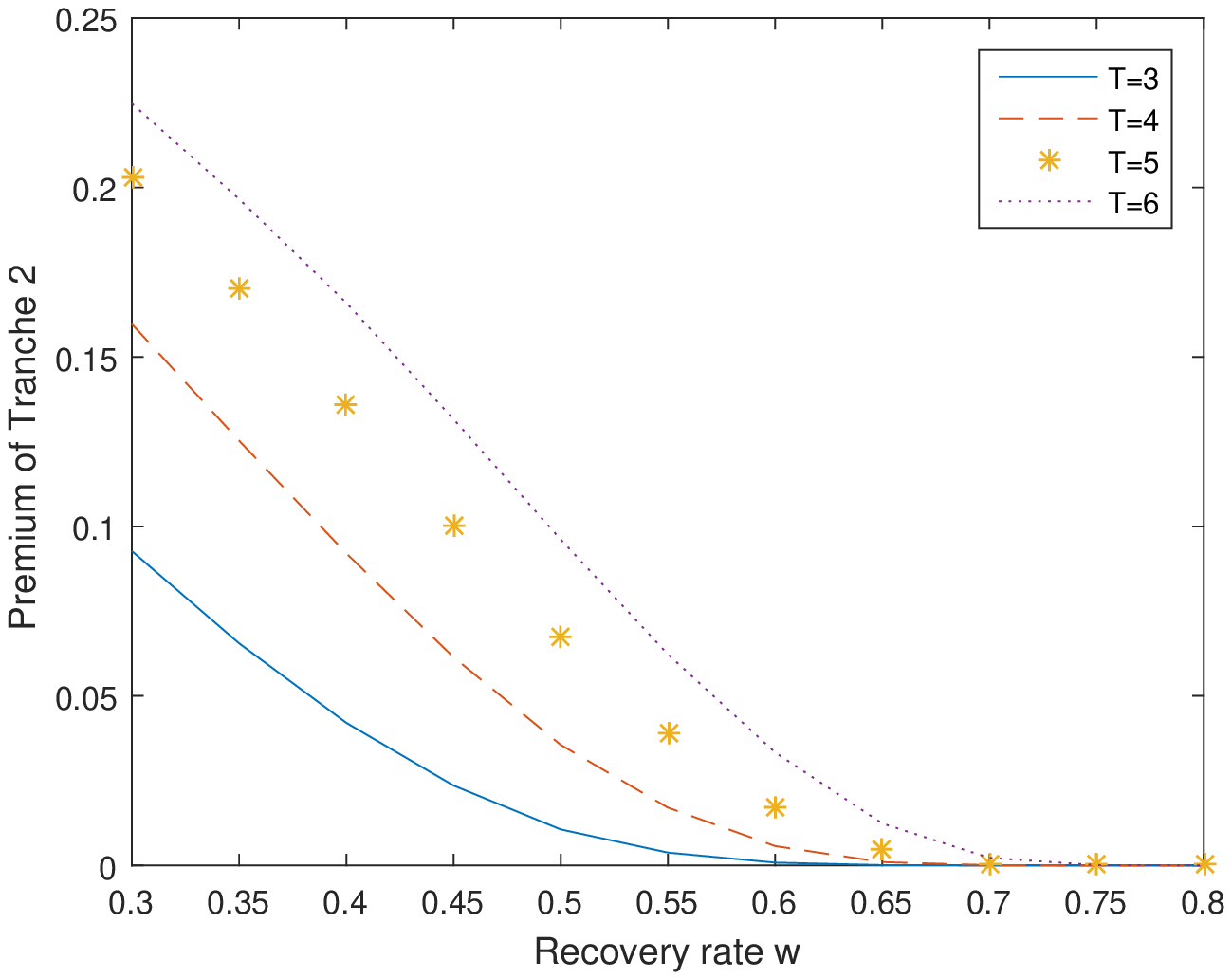}
				\label{TP2w}}
			\quad
			\subfigure[Spread of Tranche 2 (bps) against $\sigma$]{%
				\includegraphics[scale=0.4500]{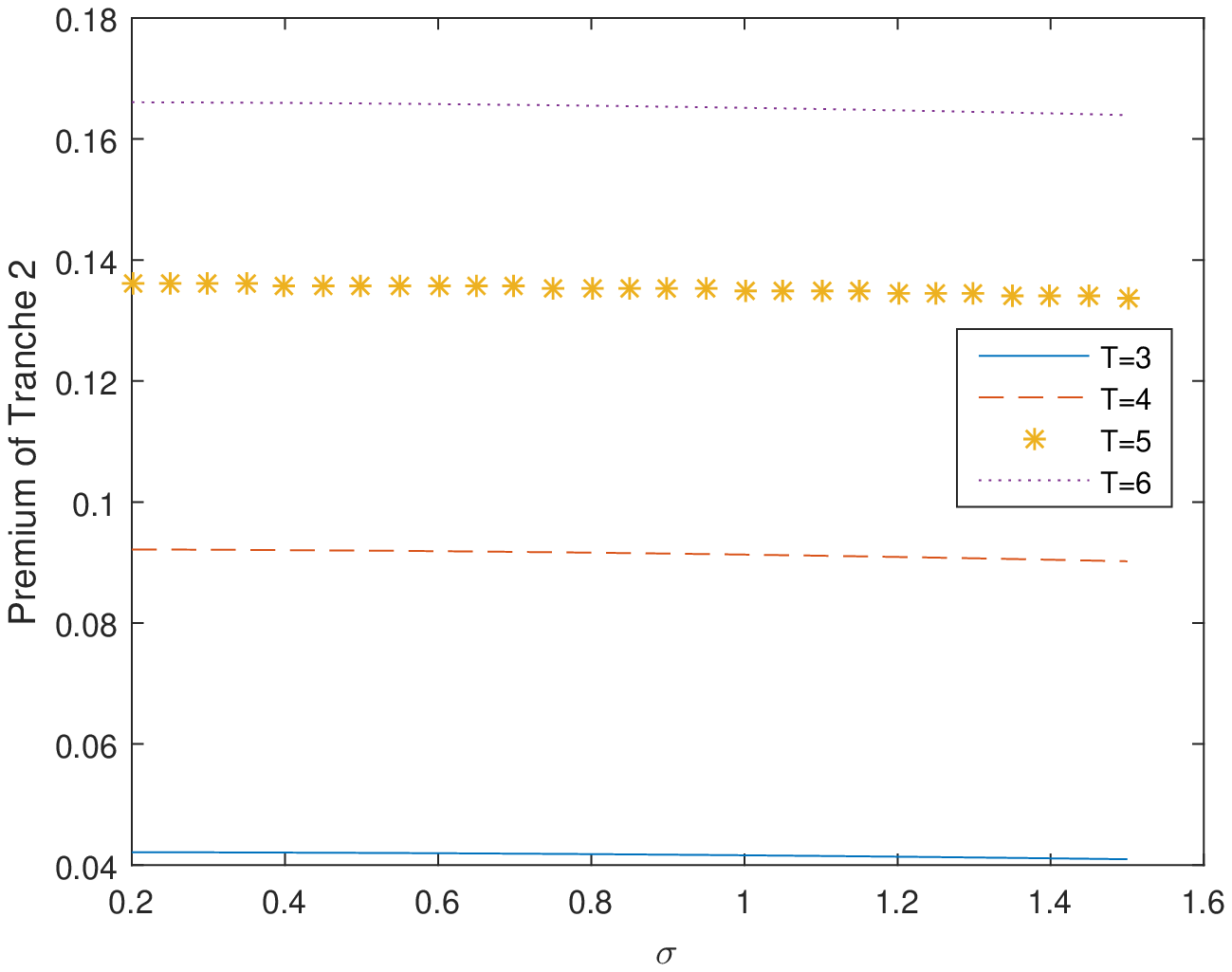}
				\label{TP2sigma}}
			\caption{Variation of Spread of Tranche 2 against different parameters governing the default of firms for different maturities for $T=3,4,5,6$ years}\label{TP2}
		\end{center}
	\end{figure}
	
	\begin{figure}[ht!]
		\begin{center}
			\subfigure[Spread of Tranche 3 (bps) against $\lambda_0$]{%
				\includegraphics[scale=0.4500]{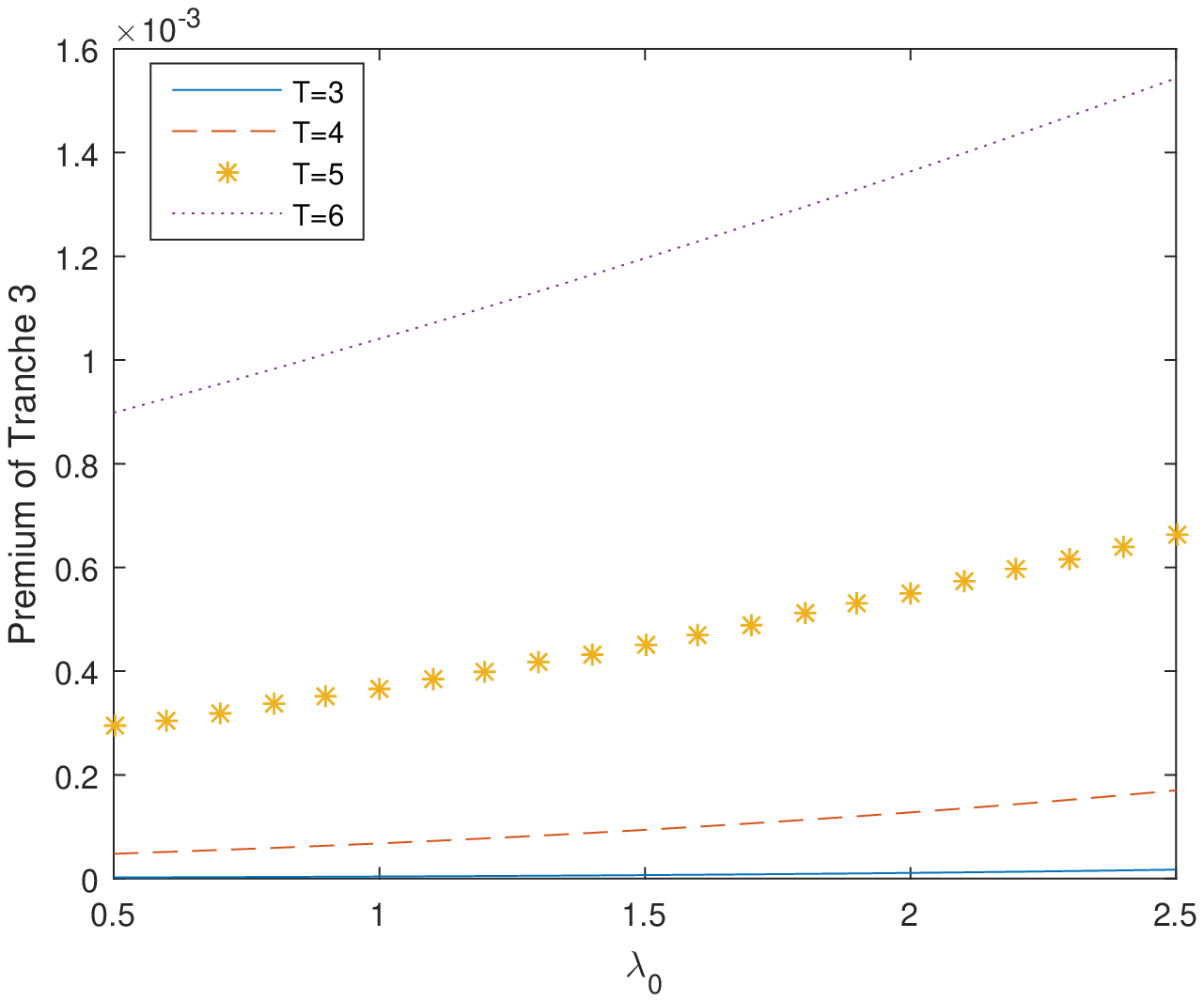}
				\label{TP3lambda}}
			\quad
			\subfigure[Spread of Tranche 3 (bps) against $\beta_y$]{%
				\includegraphics[scale=0.4500]{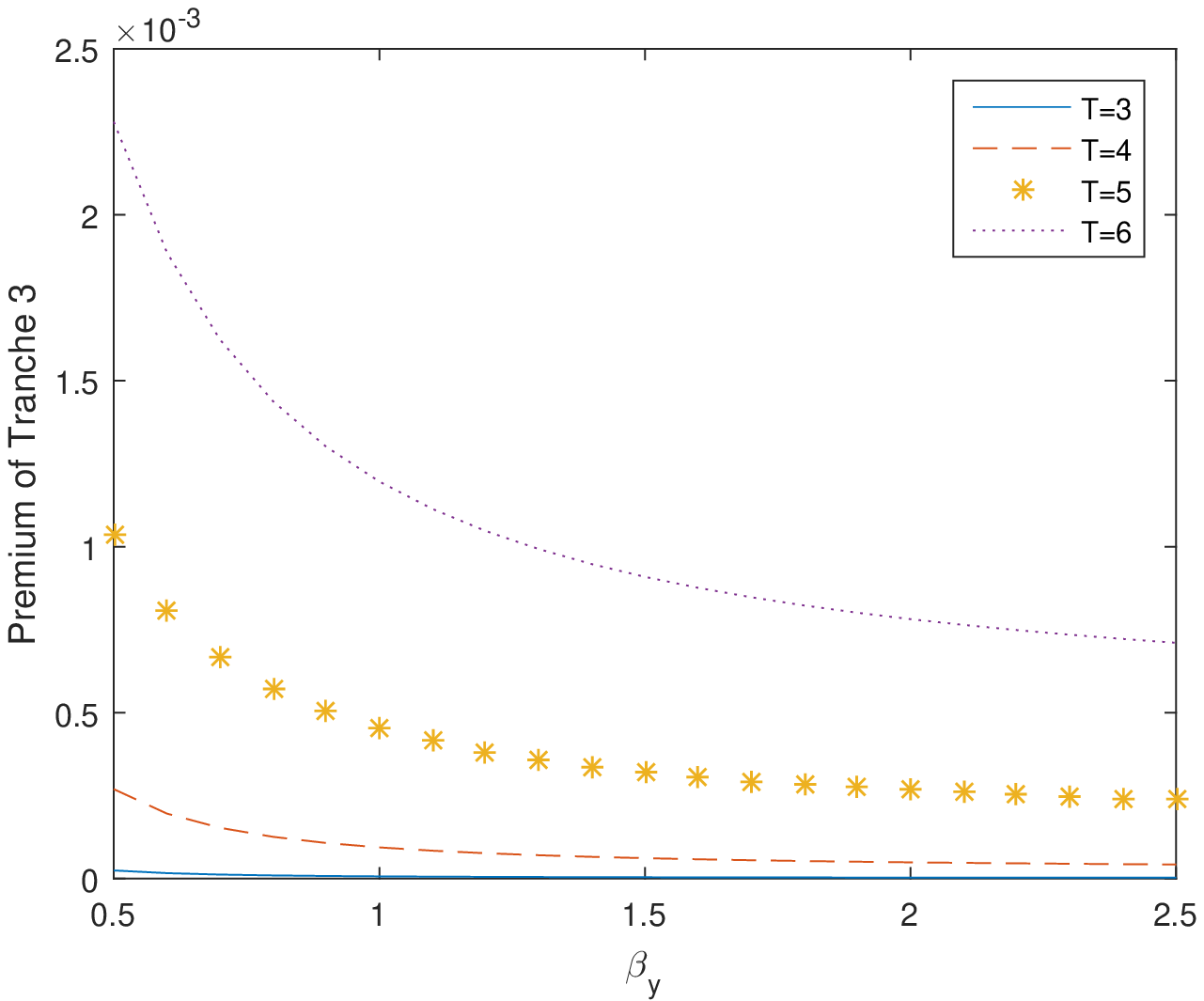}
				\label{TP3betay}}
			\quad
			\subfigure[Spread of Tranche 3 (bps) against $\eta$]{%
				\includegraphics[scale=0.4500]{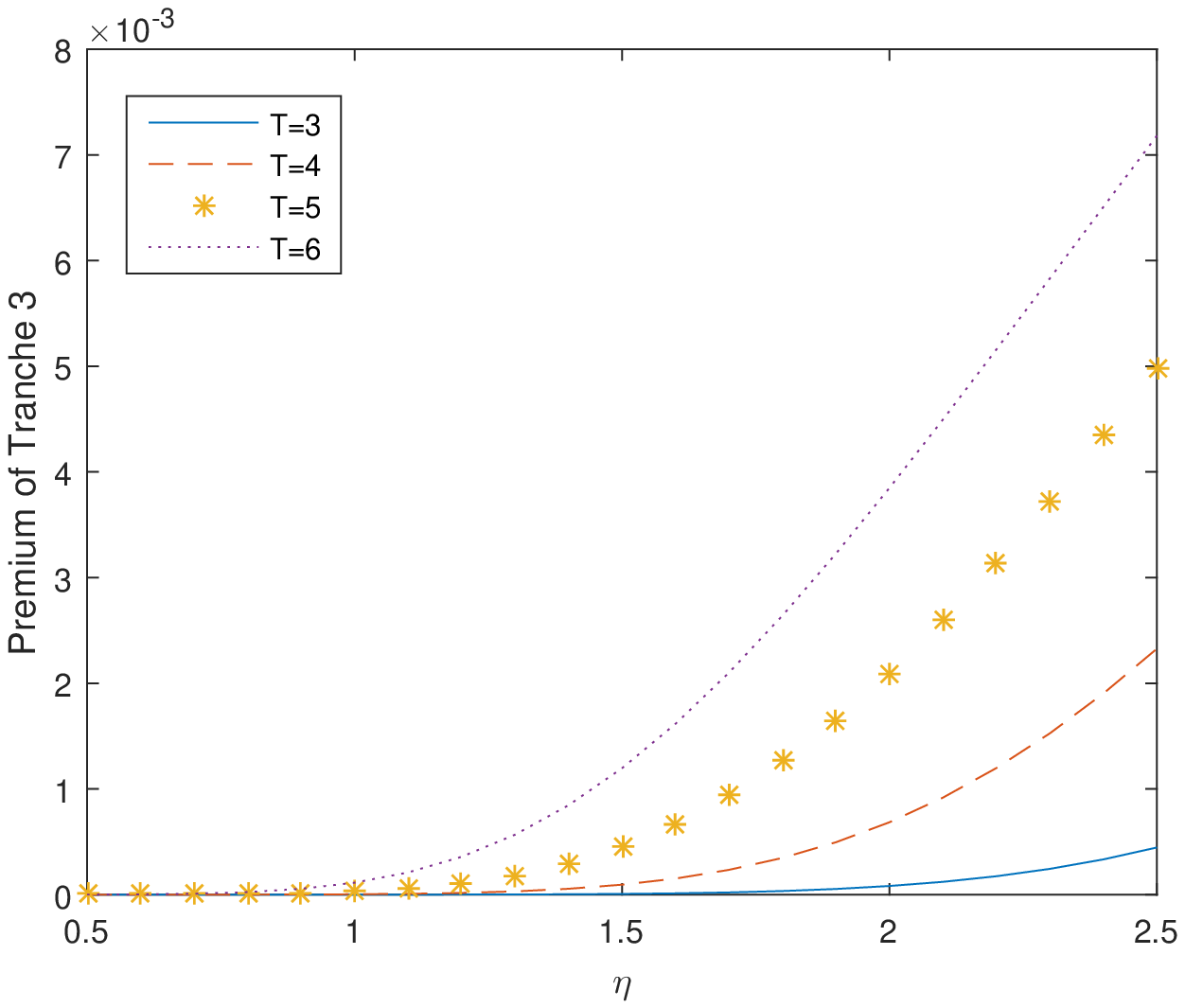}
				\label{TP3a}}
			\quad
			\subfigure[Spread of Tranche 3 (bps) against $w$]{%
				\includegraphics[scale=0.4500]{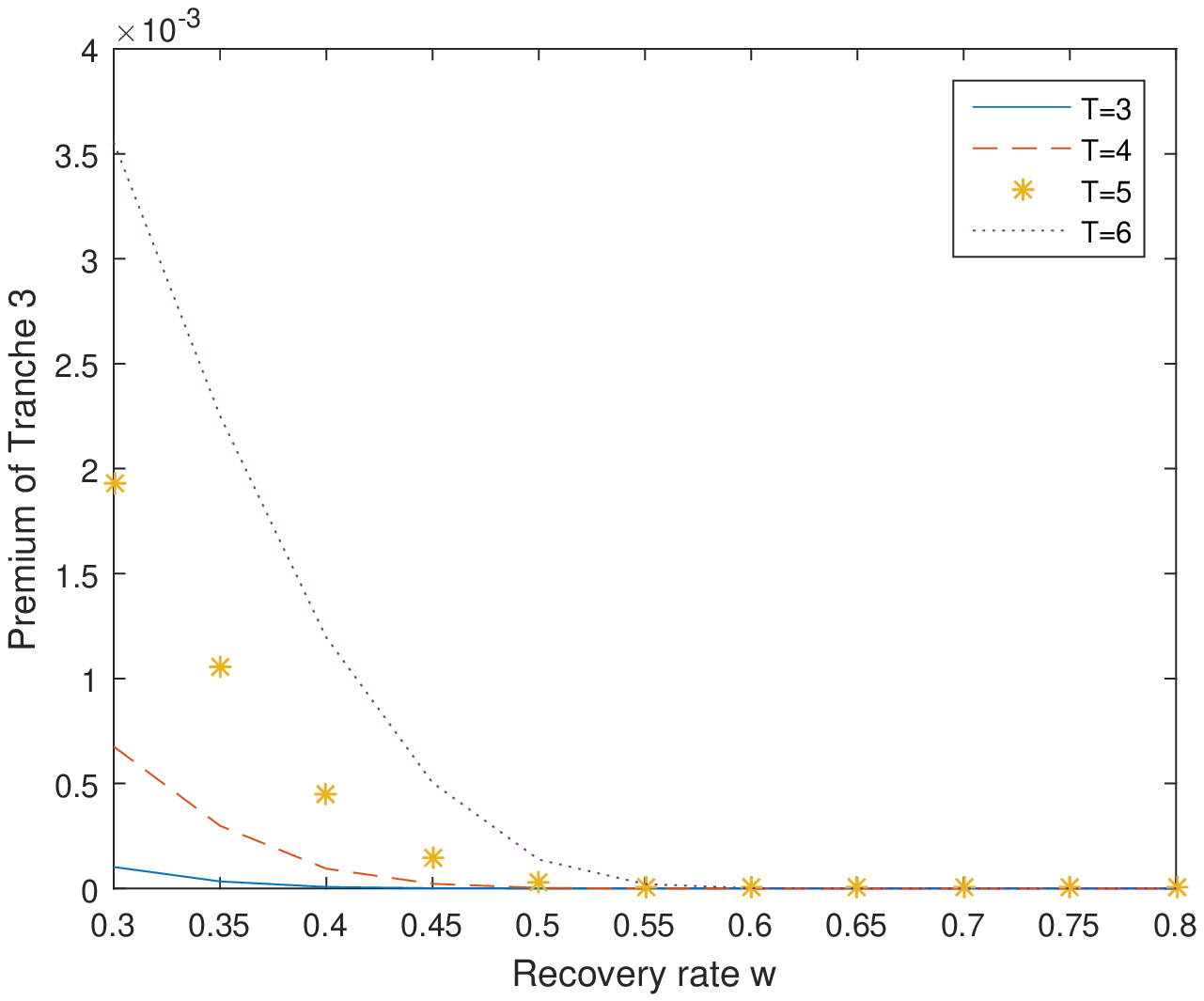}
				\label{TP3w}}
			\quad
			\subfigure[Spread of Tranche 3 (bps) against $\sigma$]{%
				\includegraphics[scale=0.4500]{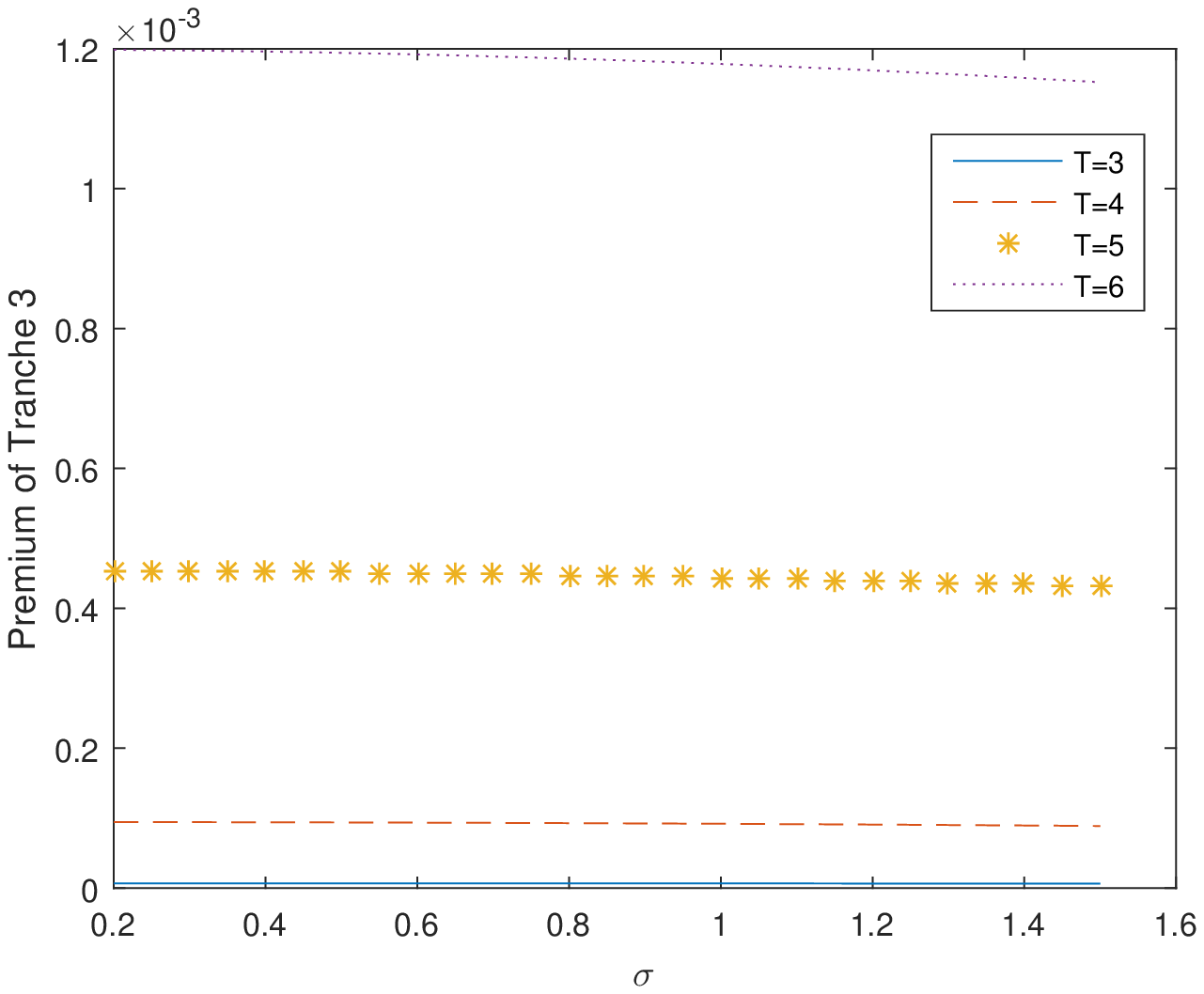}
				\label{TP3sigma}}
			\caption{Variation of Spread of Tranche 3 against different parameters governing the default of firms for different maturities for $T=3,4,5,6$ years}\label{TP3}
		\end{center}
	\end{figure}
	
\end{document}